\documentclass[11pt]{article}

\usepackage[utf8]{inputenc}
\usepackage{fullpage}
\usepackage{amsmath}
\usepackage{amsthm,amsfonts,amscd}
\usepackage{accents}
\usepackage{xcolor}
\usepackage{xspace}
\usepackage{amssymb}
\usepackage{hyperref}
\usepackage{qcircuit}
\usepackage{cleveref}
\usepackage{enumitem}

\definecolor{webgreen}{rgb}{0,.5,0}
\definecolor{webblue}{rgb}{0,0,.5}
\hypersetup{
    colorlinks,
    linkcolor={green!50!black},
    citecolor={green!50!black},
    urlcolor={blue!50!black},
    filecolor={blue!50!black},
}

\newtheorem{lemma}{Lemma}[section]
\newtheorem{theorem}[lemma]{Theorem}
\newtheorem{corollary}[lemma]{Corollary}
\newtheorem{definition}[lemma]{Definition}

\newtheorem*{remark}{Remark}
\newtheorem*{theorem-informal}{Theorem (informal)}
\theoremstyle{definition}

\newcommand{\from}{\leftarrow}

\newcommand\Tr{{\mathop\textup{Tr}}}
\newcommand{\norm}[1]{\left\|\,#1\,\right\|}
\newcommand{\trnorm}[1]{\norm{#1}_{\mathrm {tr}}}
\newcommand{\ket}[1]{\left|#1\right\rangle}
\newcommand{\bra}[1]{\left\langle #1\right|}

\newcommand{\ketbra}[2]{\ket{#1}\!\bra{#2}}

\newcommand{\kb}[1]{\ketbra{#1}{#1}}

\newcommand{\X}{\ensuremath{\mathsf{X}}\xspace}
\newcommand{\Z}{\ensuremath{\mathsf{Z}}\xspace}

\newcommand{\J}{\ensuremath{\mathcal{J}}\xspace}
\newcommand{\Jr}{\ensuremath{\mathcal{J}_\mathsf{rec}}\xspace}
\newcommand{\Crec}{\ensuremath{C_\mathsf{rec}}\xspace}
\renewcommand{\O}{\ensuremath{\mathcal{O}}\xspace}
\newcommand{\F}{\ensuremath{\mathfrak{F}}\xspace}

\newcommand{\CNOT}{\ensuremath{\mathsf{CNOT}}\xspace}

\newcommand{\reg}[1]{{\color{gray}#1}}

\newcommand{\negl}[1]{\mathsf{negl}\left(#1\right)}
\newcommand{\poly}[1]{\mathsf{poly}\left(#1\right)}

\newcommand{\advA}{\ensuremath{\mathcal{A}}\xspace}
\newcommand{\simS}{\ensuremath{\mathcal{S}}\xspace}

\newcommand{\FHE}{\ensuremath{\mathsf{FHE}}\xspace}
\newcommand{\QFHE}{\ensuremath{\mathsf{QFHE}}\xspace}
\newcommand{\KeyGen}{\ensuremath{\mathsf{KeyGen}}\xspace}
\newcommand{\BlockGen}{\ensuremath{\mathsf{BlockGen}}\xspace}
\newcommand{\SubKeyGen}{\ensuremath{\mathsf{SubKeyGen}}\xspace}
\newcommand{\Assemble}{\ensuremath{\mathsf{Assemble}}\xspace}
\newcommand{\Enc}{\ensuremath{\mathsf{Enc}}\xspace}
\newcommand{\Eval}{\ensuremath{\mathsf{Eval}}\xspace}
\newcommand{\Dec}{\ensuremath{\mathsf{Dec}}\xspace}
\newcommand{\sk}{\ensuremath{\mathit{sk}}\xspace}
\newcommand{\pk}{\ensuremath{\mathit{pk}}\xspace}

\newcommand{\pf}[1]{\ensuremath{\mathbf{P}_{#1}}\xspace}
\newcommand{\mbpf}[2]{\ensuremath{\mathbf{P}_{#1,#2}}\xspace}
\newcommand{\cc}[2]{\ensuremath{\mathbf{CC}_{#1,#2}}\xspace}
\newcommand{\mbcc}[3]{\ensuremath{\mathbf{CC}_{#1,#2,#3}}\xspace}
\newcommand{\zf}[1]{\ensuremath{\mathbf{Z}_{#1}}\xspace}

\newcommand{\ccobf}[1]{\ensuremath{O_{\mathbf{CC}}(#1)}\xspace}
\newcommand{\capprox}{\overset{c}{\approx}}
\newcommand{\params}{\mathsf{params}}
\newcommand{\inp}{\mathsf{in}}

\newcommand{\cipha}{\widetilde{\alpha}}
\newcommand{\aux}{\mathsf{aux}}
\newcommand{\key}{\mathsf{key}}
\newcommand{\ctxt}{\mathsf{ctxt}}
\newcommand{\obf}{\mathsf{obf}}
\newcommand{\obfsk}{\mathit{o}_{\sk,\beta}}
\newcommand{\Cpoint}{\ensuremath{\mathcal{C}^{\mathsf{point}}}\xspace}
\newcommand{\Czero}{\ensuremath{\mathcal{C}^{\mathsf{zero}}}\xspace}
\newcommand{\Chpoint}{\ensuremath{\hat{\mathcal{C}}^{\mathsf{point}}}\xspace}
\newcommand{\Chzero}{\ensuremath{\hat{\mathcal{C}}^{\mathsf{zero}}}\xspace}
\newcommand{\Dpoint}{\ensuremath{\mathcal{D}^{\mathsf{point}}}\xspace}
\newcommand{\Dzero}{\ensuremath{\mathcal{D}^{\mathsf{zero}}}\xspace}
\newcommand{\LWEWZ}{LWE*}

\newcommand{\qobf}[1]{\ensuremath{O_Q(#1)}\xspace}
\newcommand{\intJ}{\ensuremath{\mathcal{J}}\xspace}

\title{Impossibility of Quantum Virtual Black-Box Obfuscation of Classical Circuits}
\author{Gorjan Alagic\footnote{QuICS, University of Maryland; National Institute of Standards and Technology} \and
	\setcounter{footnote}{2} %
	Zvika Brakerski\footnote{Weizmann Institute of Science, \texttt{zvika.brakerski@weizmann.ac.il}. } \and
	Yfke Dulek\footnote{QuSoft; University of Amsterdam,
          \texttt{yfkedulek@gmail.com, c.schaffner@uva.nl}} \and
	Christian Schaffner$^\fnsymbol{footnote}$}
\date{\vspace{-1cm}}

\begin{document}

\maketitle

\begin{abstract}
  Virtual black-box obfuscation is a strong cryptographic primitive: it encrypts a circuit while maintaining its full input/output functionality. A remarkable result by Barak et al. (Crypto 2001) shows that a general obfuscator that obfuscates classical circuits into classical circuits cannot exist. A promising direction that circumvents this impossibility result is to obfuscate classical circuits into quantum states, which would potentially be better capable of hiding information about the obfuscated circuit. We show that, under the assumption that Learning With Errors (LWE) is hard for quantum computers, this quantum variant of virtual black-box obfuscation of classical circuits is generally impossible. On the way, we show that under the presence of dependent classical auxiliary input, even the small class of classical point functions cannot be quantum virtual black-box obfuscated.
\end{abstract}

\section{Introduction}
The obfuscation of a circuit is an object, typically another circuit, that allows a user to evaluate the functionality of the original circuit without learning any additional information about the structure of the circuit. Obfuscation is useful for publishing software without revealing the code, but it also has more fundamental applications in cryptography. For example, the strongest notion called \emph{virtual black-box} obfuscation can transform any private-key encryption scheme into a public-key scheme, and transform public-key schemes into fully-homomorphic schemes. Unfortunately, this notion turns out to be impossible for general circuits~\cite{BGI+01} -- at least, if we require the obfuscation of a circuit to be a circuit itself.

The impossibility result from~\cite{BGI+01} leaves open an intriguing possibility: what if the obfuscation of a (classical) circuit is allowed to be a \emph{quantum state}? Could a quantum state contain all the information about a functionality, allowing a user to produce correct outputs, without revealing all that information? This possibility seems hopeful, due to the unrevealing nature of quantum states. However, in this work, we show that virtual-black-box obfuscating classical circuits into quantum states is not possible. 
\subsection{Related work}\label{sec:intro-related-work}
Barak et al.\ defined the obfuscating property of virtual black-box (vbb) obfuscators as follows: any information that an adversary can learn about a circuit from its obfuscation can also be learned by a simulator that does not have access to the obfuscation, but only to an oracle for the circuit's functionality~\cite{BGI+01}. In this definition, the crucial difference between the adversary and the simulator is that the adversary has access to a short representation of the circuit (namely, the obfuscation), whereas the simulator only has access to an input/output interface that implements the functionality. Some circuit classes allow the adversary to exploit this difference by using the obfuscation as an input value to the circuit itself. Those circuit classes are unobfuscatable in the vbb sense, rendering vbb obfuscation impossible for the general class of circuits in $\mathbf{P}$~\cite{BGI+01}.

In more detail, the impossibility proof in~\cite{BGI+01} relies on point functions, which output zero everywhere except at a single input value $\alpha$, where they output a string $\beta$. The circuits in the unobfuscatable class can, depending on the input, do all of the following: (1) apply that point function, (2) return an encryption of $\alpha$, (3) homomorphically evaluate a gate, or (4) check whether a ciphertext decrypts to $\beta$. An adversary holding the obfuscation is able to divide it into single gates, and can use those to homomorphically evaluate option (1), thereby converting a ciphertext for $\alpha$ into a ciphertext for $\beta$. That way, the adversary can tell whether he is holding an obfuscation with a point function from $\alpha$ to $\beta$, or one with the all-zero function. (In the second case, the homomorphic evaluation would yield a ciphertext for zero, rather than one for $\beta$.) A simulator, only having access to the input/output behavior, cannot perform the homomorphic evaluation, because it cannot divide the functionality into single gates.

The above construction rules out the existence of an obfuscator that maps classical circuits to classical circuits. It leaves open the possibility of an obfuscator that maps classical circuits to \emph{quantum states}: such a quantum state, together with a fixed public `interpreter map', could be used to evaluate the obfuscated circuit.
The possibility of quantum obfuscation was the object of study for Alagic and Fefferman~\cite{AF16}, who attempted to port the impossibility proof from~\cite{BGI+01} to the quantum setting. In doing so, they encountered two issues:
\begin{description}[style=unboxed,leftmargin=0.5cm]
	\item[Homomorphic evaluation.] The interpreter map, that runs the obfuscation state on a chosen input, is a quantum map. It will likely have quantum states as intermediate states of the computation, so in order to homomorphically run the point function, one needs the ability to evaluate quantum gates on quantum ciphertexts. The unobfuscatable circuit class will thus need to contain quantum circuits to perform homomorphic evaluation steps.
	\item[Reusability.] In the construction from~\cite{BGI+01}, the obfuscated circuit needs to be used multiple times: for example, each homomorphic gate evaluation requires a separate call to the obfuscated circuit. If the obfuscation is a (classical or quantum) circuit, this poses no problem, but if it is a quantum \emph{state}, multiple uses are not guaranteed.
\end{description}
These two issues limit the extent of the impossibility results in~\cite{AF16}: they show that it is impossible to vbb obfuscate \emph{quantum} circuits into \emph{reusable} obfuscated states (e.g., quantum circuits).

After it became clear~\cite{BGI+01} that obfuscating all classical circuits is impossible, efforts were made to construct obfuscators for smaller, but still nontrivial, classes of circuits. Successful constructions have been found for several classes of evasive functions, such as point functions~\cite{Wee05,CD08} and compute-and-compare functions~\cite{WZ17,GKW17}. Currently, no quantum obfuscators are known for circuit classes that cannot be classically obfuscated.

\subsection{Our contributions}\label{sec:intro-contributions}
We strengthen the impossibility of virtual-black-box obfuscation of classical circuits by showing that classical circuits cannot be obfuscated into quantum states. We assume the existence of classical-client quantum fully homomorphic encryption and classical obfuscation of compute-and-compare functions. Both of these can be constructed from the learning-with-errors (LWE) assumption~\cite{Mahadev18,Brakerski18,WZ17,GKW17}. The compute-and-compare construction requires the strongest assumption in terms of the LWE parameters.

\begin{theorem-informal}
	If LWE is hard for quantum algorithms, then it is impossible to quantum vbb obfuscate the class of polynomial-size classical circuits (even with non-negligible correctness and security error, and even if the obfuscation procedure is inefficient).
\end{theorem-informal}

Our result uses the same proof strategy as in~\cite{BGI+01} and~\cite{AF16}, overcoming the two main issues described above as follows:
\begin{description}[style=unboxed,leftmargin=0.5cm]
	\item[Homomorphic evaluation.] The constructions in~\cite{BGI+01} and ~\cite{AF16} rely on the obfuscator to implement the homomorphic evaluations, by obfuscating the functionality ``decrypt, then apply a gate, then re-encrypt". However, by now, we know how to build quantum fully-homomorphic encryption schemes directly~\cite{Mahadev18,Brakerski18}, based on the learning-with-errors (LWE) assumption. Thus, in our construction, we can remove the homomorphic gate evaluation from the obfuscated circuits: the adversary can do the homomorphic evaluation of the point function herself, using a quantum fully-homomorphic encryption scheme. With the homomorphic evaluation removed from it, the class of circuits that we prove impossible to obfuscate can remain classical.
	
	This solution introduces a slight complication: part of the functionality of the circuit we construct is now to return the public evaluation key. However, unless one is willing to make an assumption on the circular security of the homomorphic encryption, the size of this key (and therefore the size of the circuit) scales with the size of the circuit that needs to be homomorphically evaluated. To get rid of this inconvenient dependence, our unobfuscatable circuit returns the public key in small, individual blocks that can be independently computed. We argue that any classical-key quantum fully-homomorphic encryption scheme has public keys that can be decomposed in this way.
	
	\item[Reusability.] The circuits that we consider are classical and deterministic. Therefore, if the interpreter map is run on an obfuscation state $\rho$ for a circuit $C$, plus a classical input $x$, then by correctness, the result is (close to) a computational-basis state $\ket{C(x)}$. This output can be copied out to a separate wire without disturbing the state, and the interpreter map can be reversed, recovering the obfuscation $\rho$ to be used again. If the interpreter map is not unitary, then it can be run coherently (i.e., keeping purification registers around instead of measuring wires), and this coherent version can be reversed as long as the purification registers are not measured.
	
	At one point in our proof, we will need to run the interpreter map homomorphically on (an encryption of) $\rho$ and $x$. This may result in a superposition of different ciphertexts for $C(x)$, which cannot cleanly be copied out to a separate wire without entangling that wire with the output. Thus, recovering $\rho$ is not necessarily possible after the homomorphic-evaluation step.
	
	We circumvent this problem by making sure that the homomorphic evaluation occurs last, so that $\rho$ is not needed anymore afterwards. This reordering is achieved by classically obfuscating the part of the circuit that checks whether a ciphertext decrypts to the value $\beta$. That way, this functionality becomes a constant output value that a user can request and store before performing the homomorphic evaluation, and use afterwards. To obfuscate the decryption check, we use a classical vbb obfsucator for compute-and-compare functions, which relies on a variant of the LWE assumption~\cite{WZ17,GKW17}.
\end{description}
Our impossibility result compares to the classical impossibility result from~\cite{BGI+01} as follows.
First, as mentioned, we extend the realm of impossible obfuscators to include obfuscators that produce a quantum state, rather than a classical circuit.
Second, the impossibility result from~\cite{BGI+01} is unconditional, whereas we require the (standard) assumption that learning-with-errors is hard for quantum adversaries. It may be possible to relax this assumption if $\rho$ can be recovered after the homomorphic evaluation, see Section~\ref{sec:intro-open-questions} below.
Third, the class of classical circuits that cannot be obfuscated is slightly different: in our work, it does not have the homomorphic-evaluation functionality built into it, and is therefore arguably simpler, strengthening the impossibility result. However, we stress that in both works, the unobfuscatable circuit class itself is somewhat contrived: the main implication is that its superclass $\mathbf{P}$ is unobfuscatable.

As an intermediate result, we show that it is impossible to vbb obfuscate even just the class of classical multi-bit-output point functions into a quantum state, if the adversary and simulator have access to auxiliary classical information that contains an encryption of the non-zero input value $\alpha$ and a vbb obfuscation of a function depending on the secret key for that encryption.

\begin{theorem-informal}
	If LWE is hard for quantum algorithms, then it is impossible to quantum vbb obfuscate multi-bit-output point functions and the all-zero function under the presence of classical dependent auxiliary information (even with non-negligible soundness and security error).
\end{theorem-informal}

At first glance, that may seem to contradict the constructions in~\cite{WZ17,GKW17}, where vbb obfuscation for point functions is constructed, even in the presence of dependent auxiliary information. The crucial difference is that the~\cite{WZ17,GKW17} constructions only allow a limited dependency of the auxiliary information, whereas in our impossibility proof, the dependence is slightly stronger. This subtle difference seems to indicate that the gap between possibility and impossibility of vbb obfuscation is closing.

\paragraph{Comparison with Concurrent Work.} 
Independently of this work, Ananth and La Placa~\cite{AL20} have concurrently shown the general impossibility of quantum copy-protection, thereby also ruling out quantum obfuscation of classical circuits. Their techniques are very similar to ours, but their adversary is somewhat more powerful in the sense that it allows to completely de-obfuscate the program given non-black-box access. They also present some positive results in their work.
Their result relies on the same LWE assumption as ours, but in addition they require that the underlying homomorphic-encryption scheme is circularly secure. We avoid circularity by introducing a notion of decomposable public keys for homomorphic encryption. This technique could potentially be used to remove the circularity assumption from the copy-protection impossibility result~\cite{AL20} as well.

\subsection{Open questions}\label{sec:intro-open-questions}
The strongest assumption in our work is the existence of the classical vbb obfuscator for compute-and-compare functions, which relies on a variant of LWE. It is necessary because the QFHE evaluation may destroy the obfuscation state when the superposition of output ciphertexts is measured. However, it is not clear if this measurement actually destroys any information on the \emph{plaintext} level, since the plaintext value is deterministic. Thus, it may be possible to recover the (plaintext) obfuscation state after the QFHE evaluation. In that case, it is not necessary to classically obfuscate the compute-and-compare function: it can simply be part of the quantum-obfuscated functionality.

Other open questions are about possibilities rather than impossibilities. What circuit classes \emph{can} be vbb obfuscated into quantum states? Is quantum vbb obfuscation stronger than classical vbb obfuscation, in the sense that it can obfuscate circuit classes that classical vbb cannot? Also, the weaker notion of indistinguishability obfuscation (iO) (also introduced in~\cite{BGI+01}) is not affected by our impossibility result: it may still be possible to classically or quantumly iO obfuscate classical functionalities. Could such a construction be lifted into the quantum realm, so that we can (quantum) iO obfuscate quantum functionalities?~\cite{BK20}

\subsection{Structure of this work}
In \Cref{sec:preliminaries}, we give preliminary definitions of the relevant concepts for this work: (classical and quantum) obfuscation, quantum fully homomorphic encryption, and compute-and-compare functions. We also describe how the input of an (almost) deterministic quantum circuit can be recovered. In \Cref{sec:aux}, we prove impossibility of quantum obfuscation of point functions under dependent auxiliary input. Building on the concepts in that section, \Cref{sec:no-aux} proves our main result, impossibility of quantum obfuscation of classical circuits without any auxiliary input.

\section{Preliminaries}\label{sec:preliminaries}

\subsection{Notation}\label{sec:prelim-notation}
PPT stands for probabilistic polynomial-time algorithm, and QPT stands for quantum polynomial-time algorithm. If a classical or quantum algorithm $A$ has oracle access to a classical function $f$, we write $A^f$. If $A$ has access to multiple oracles with separate input/output interfaces, we write, e.g.,  $A^{f,g}$. Since our oracles will model or emulate evaluation of a known circuit, our quantum algorithms will always have superposition access to a classical oracle $f$. This amounts to oracle access to the unitary map $\ket{x}\ket{y} \mapsto \ket{x}\ket{y \oplus f(x)}$ where $\oplus$ is the bit-wise XOR operation.

Let $\poly{x}$ denote an unspecified polynomial $p(x)$. Similarly, let $\negl{x}$ denote an unspecified negligible function $\mu(x)$, i.e., for all constants $c \in \mathbb{N}$ there exists an $x_0 \in \mathbb{R}$ such that for all $x > x_0$, $|\mu(x)| < x^{-c}$. Let $\zf{n} :\{0,1\}^n \to \{0^n\}$ denote the all-zero function on $n$ input bits: $\zf{n}(x) = 0^n$ for all $x$.

If $D$ is a distribution, we write $x \from D$ to signify that $x$ is sampled according to $D$. For a finite set $S$, we write $x \from_R S$ to signify that $x$ is sampled uniformly at random from the set $S$. Two distribution ensembles $\{D_{\lambda}\}_{\lambda \in \mathbb{N}}$ and $\{D'_{\lambda}\}_{\lambda \in \mathbb{N}}$ are computationally indistinguishable, written $D_{\lambda} \capprox D'_{\lambda}$, if no poly-time algorithm can distinguish between a sample from one distribution or the other, i.e., for all PPT $A$,
\[
\left|
\Pr_{x \from D_{\lambda}}[A(x) = 1]
-
\Pr_{y \from D_{\lambda}'}[A(y) = 1]
\right|
\leq
\negl{\lambda}.
\]
We sometimes write $x \capprox y$ if it is clear from which distributions $x$ and $y$ are sampled. If not even a QPT algorithm can distinguish them, the distributions are quantum-computationally indistinguishable.

A pure quantum state is written $\ket{\psi}$ or $\ket{\varphi}$, and a mixed quantum state is usually denoted by $\rho$ or $\sigma$. As a special case, a computational-basis state is written $\ket{x}$ for some classical string $x \in \{0,1\}^*$. We sometimes abuse notation and give a classical input $x$ to a quantum algorithm $A$, writing $A(x)$: in that case, the algorithm $A$ is actually given $\ket{x}$ as input.

$\X$ and $\Z$ denote the bit-flip gate and phase-flip gate, respectively. If we write $\X^a$ for some $a \in \{0,1\}$, we mean that the gate $\X$ is applied if $a = 1$; otherwise, identity is applied.

Finally, for a mixed state $\rho$, let $\trnorm{\rho} := \Tr\left(\sqrt{\rho^{\dagger}\rho}\right)$ denote the trace norm. The trace distance $\frac{1}{2}  \trnorm{\rho - \sigma}$ is a measure for how different two mixed states $\rho$ and $\sigma$ are.

\subsection{Classical and Quantum Virtual-Black-Box Obfuscation}\label{sec:prelim-obfuscation}
In this work we consider so-called \emph{circuit} obfuscators: the functionalities to be hidden are represented by circuits. A virtual-black-box circuit obfuscator hides the functionality in such a way that the obfuscation looks like a ``black box": the only way to get information about its functionality is to evaluate it on an input and observe the output.

\begin{definition}[{\cite[Definition 2.2]{BGI+01}}]\label{def:classical-vbb}
	A classical virtual black-box obfuscator for the circuit class \F is a probabilistic algorithm \O such that
	\begin{enumerate}
		\item (polynomial slowdown) For every circuit $C \in \F$, $|\O(C)| = \poly{|C|}$;
		\item (functional equivalence) For every circuit $C \in \F$, the string $\O(C)$ describes a circuit that computes the same function as $C$;
		\item (virtual black-box) For any PPT adversary $A$, there exists a PPT simulator $\simS$ such that for all circuits $C \in \F$,
		\[
		\left|
		\Pr\left[
		A(\O(C)) = 1
		\right]
		-
		\Pr\left[
		A(\simS^{C}(1^{|C|})) = 1
		\right]
		\right|
		\leq
		\negl{|C|}.
		\]
	\end{enumerate}
\end{definition}
As a variation on the third requirement, one may assume that some auxiliary information (which may depend on the circuit $C$) is present alongside the obfuscation $\O(C)$. In that case, a simulator with access to that auxiliary information should still be able to simulate the adversary's output distribution:
\begin{definition}[{\cite[Definition 3]{GK05}}]\label{def:classical-vbb-aux}
	A classical virtual black-box obfuscator w.r.t.\ dependent auxiliary input for a circuit class \F is a probabilistic algorithm \O that satisfies \Cref{def:classical-vbb}, with the ``virtual black-box" property redefined as follows:
	\begin{enumerate}
		\setcounter{enumi}{2}
		\item (virtual black-box) For any PPT adversary $A$, there exists a PPT simulator $\simS$ such that for all circuits $C \in \F$ and all strings $\aux \in \{0,1\}^{\poly{|C|}}$ (which may depend on $C$),
		\[
		\left|
		\Pr\left[
		A(\O(C), \aux) = 1
		\right]
		-
		\Pr\left[
		A(\simS^{C}(1^{|C|},\aux)) = 1
		\right]
		\right|
		\leq
		\negl{|C|}.
		\]
	\end{enumerate}
\end{definition}

In the quantum setting, we consider quantum obfuscators for classical circuit classes: that is, the obfuscation $\O(C)$ may be a quantum state. We adapt Definition 5 from~\cite{AF16}, which defines quantum obfuscators for quantum circuits.

\begin{definition}\label{def:quantum-vbb}
	A quantum virtual black-box obfuscator for the classical circuit class \F is a quantum algorithm \O and a QPT \J such that
	\begin{enumerate}
		\item (polynomial expansion) For every circuit $C \in F$, $\O(C)$ is an $m$-qubit quantum state with $m = \poly{n}$;
		\item (functional equivalence) For every circuit $C \in F$ and every input $x$, \[\frac{1}{2}\trnorm{\J(\O(C) \otimes \kb{x}) - \kb{C(x)}} \leq \negl{|C|};\]
		\item (virtual black-box) For every QPT adversary $\advA$, there exists a QPT simulator $\simS$ (with superposition access to its oracle) such that for all circuits $C \in \F$,
		\[
		\left| \Pr[ \advA (\O(C)) = 1] - \Pr[ \simS^{C}(1^{|C|}) = 1] \right| \leq \negl{|C|} .
		\]
	\end{enumerate}
\end{definition}
	There are a few differences with the classical definition. First, the obfuscation is a quantum state, and not a (classical or quantum) circuit. Second, due to the probabilistic nature of quantum computation, we allow a negligible error in the functional equivalence. Third, the simulator is slightly more powerful because of its superposition access to the functionality of $C$: a query performs the unitary operation specified by $\ket{x}\ket{z} \mapsto \ket{x}\ket{z \oplus C(x)}$. Note that a quantum adversary can always use a (classical or quantum) obfuscation to compute the obfuscated functionality on a superposition of inputs, obtaining a superposition of outputs. For this reason, the simulator gets superposition access to its oracle in the quantum setting. Throughout this work, all oracles supplied to quantum algorithms allow for superposition access.
	
	We can again strengthen the virtual black-box property to include (classical or quantum) dependent auxiliary information: this auxiliary string or state would be provided to both the adversary and the simulator, in the same way as in \Cref{def:classical-vbb-aux}.

\subsection{Quantum Fully Homomorphic Encryption}\label{sec:prelim-qfhe}
A fully homomorphic encryption (FHE) of a message $m$ provides privacy by hiding the message, but allows ciphertexts to be transformed in a meaningful way. Given a ciphertext for $m$, some party that only knows the public key can produce a ciphertext for $f(m)$ for any efficiently computable function $f$. Any information that is necessary for this transformation is contained in the public key (in particular, we do not make a distinction between the public key and the evaluation key).

A \emph{quantum} fully homomorphic encryption (QFHE) scheme allows quantum computations on encrypted quantum data. From the Learning with Errors assumption, it is possible to construct secure QFHE schemes where all client-side operations (key generation, encryption, and decryption) are classical~\cite{Mahadev18,Brakerski18}.

\begin{definition}\label{def:qfhe}
A quantum fully homomorphic encryption scheme \QFHE consists of four algorithms, as follows:
\begin{itemize}
    \item \textbf{Key Generation:} $(\pk, \sk) \leftarrow \QFHE.\KeyGen(1^\lambda)$ produces a public key \pk and a secret key \sk, given a security parameter $\lambda$. This is a classical PPT algorithm.
    \item \textbf{Encryption:} $c \leftarrow \QFHE.\Enc_{\pk}(m)$ encrypts a single-bit message $m \in \{0,1\}$. For multi-bit messages $m \in \{0,1\}^\ell$, we write $\QFHE.\Enc_{\pk}(m)$ to denote the bit-by-bit encryption 
    $$
    (\QFHE.\Enc_{\pk}(m_1), \QFHE.\Enc_{\pk}(m_2), \dots, \QFHE.\Enc_{\pk}(m_\ell))\,.
    $$
    This algorithm is in general QPT but it only uses a classical random tape, and furthermore whenever $m$ is classical, so is the encryption algorithm.
    \item \textbf{Decryption:} $m' = \QFHE.\Dec_{\sk}(c)$ decrypts a ciphertext $c$ into a single-bit message $m'$, using the secret key \sk. If $c$ is a ciphertext for a multi-bit message, we write $\QFHE.\Dec_{\sk}(c)$ for the bit-by-bit decryption. Again this is QPT in general, but can be classical if $c$ is classical.
    \item \textbf{Homomorphic evaluation:} $c' \from \QFHE.\Eval_{\pk}(C, c)$ takes as input the public key, a classical description of a BQP circuit $C$ with $\ell$ input wires and $\ell'$ output wires, and a bit-by-bit encrypted ciphertext $c$ encrypting $\ell$ bits. It produces a $c'$, a sequence of $\ell'$ output ciphertexts. This is a QPT algorithm.
\end{itemize}
\end{definition}

We say that a (Q)FHE scheme is (perfectly) correct if the homomorphic evaluation of any BQP circuit $C$ on a ciphertext has the effect of applying $C$ to the plaintext, i.e.,
\[
\QFHE.\Dec_{\sk}\left(\QFHE.\Eval_{\pk}\left(C,\QFHE.\Enc_{\pk}(m)\right)\right) = C(m)
\]
for all $m$, $C$, and $(\pk,\sk) \from \QFHE.\KeyGen(1^{\lambda})$. A (Q)FHE scheme is secure if its encryption function is secure. We usually require quantum indistinguishability under chosen plaintext attacks (q-IND-CPA)~\cite{BJ15}.

The QFHE schemes from~\cite{Mahadev18,Brakerski18} encrypt a message $m$ using a quantum one-time-pad with random keys $a, b \in \{0,1\}$, attaching classical FHE ciphertexts of the one-time pad keys:
\[\QFHE.\Enc_{pk}(m) = \X^a\Z^b\ket{m} \otimes \ket{\FHE.\Enc_{\pk}(a), \FHE.\Enc_{\pk}(b)}.
\]
Note that this ciphertext can be classically represented as the tuple 
$$
(m \oplus a, \FHE.\Enc_{\pk}(a), \FHE.\Enc_{\pk}(b))\,,
$$
so that encryption may be seen as a classical procedure. Conversely, a classical homomorphic encryption $\tilde{m} \from \FHE.\Enc_{\pk}(m)$ can easily be turned into a valid quantum homomorphic encryption by preparing the state
$
\ket{0} \otimes \ket{\tilde{m}, \FHE.\Enc_{\pk}(0)},
$
which decrypts to $m$. Thus, it is possible to freely switch back and forth between quantum ciphertexts and classical ciphertexts, as long as the message is known to be classical.

The quantum one-time pad encryption also straightforwardly extends to encrypting general quantum states $\ket{\psi}$, rather than only computational-basis states $\ket{m}$: the quantum one-time pad is simply applied to the state $\ket{\psi}$, and the one-time-pad keys encrypted into a computational-basis state as above. Of course, encryption becomes a quantum procedure in this setting. We will use encryption of quantum states in our work, where we supply the encryption of a quantum-state obfuscation as the input to a homomorphic evaluation.

\paragraph{Leveled FHE and Bootstrappable FHE.} In many cases in the literature, we wish to consider FHE schemes which require an a priori upper bound (polynomial in the security parameter) on the depth of circuits to be homomorphically evaluated. In the current state of the art, such schemes (referred to as \emph{leveled} FHE) can be constructed under milder assumptions than unleveled schemes: in particular, they do not require circular-security-type assumptions. There are a few variants of leveled FHE defined in the literature, but for the purpose of this work we use the following.

\begin{definition}[Leveled FHE]\label{def:levfhe}
	A leveled (Q)FHE scheme is a scheme where the key generation takes an additional parameter $\KeyGen(1^{\lambda}, 1^d)$ and outputs $(\sk, \pk)$ for a (Q)FHE scheme. Correctness holds only for evaluating circuits of total depth at most $d$. Furthermore, the length of $\sk$ and the complexity of decryption are independent of $d$.	
\end{definition}
We assume w.l.o.g.\ that the random tape used by $\KeyGen$ is always of length $\lambda$ and does not depend on $d$ (this is w.l.o.g.\ since it is always possible to use a PRG to stretch the random tape into the desired length).

One way to construct leveled FHE is via the bootstrapping technique as suggested by Gentry~\cite{Gen09}. Gentry showed that given a base scheme with homomorphic capacity greater than its decryption depth, it is possible to create a leveled scheme with the following properties.
\begin{definition}[Leveled Bootstrapped FHE]\label{def:bootfhe}
A leveled bootstrapped (Q)FHE is one where there exists a base-scheme with a key-generation algorithm $\SubKeyGen(1^\lambda)$ and encryption, decryption and evaluation algorithms, such that the key generation algorithm $\KeyGen(1^{\lambda}, 1^d)$ takes the following form. 
\begin{enumerate}
	\item Run $\SubKeyGen(1^\lambda)$ with fresh randomness $(d+1)$ times to generate sub-keys $(\sk_i, \pk_i)$ for $i=0, \ldots, d$.
	
	\item Encrypt $c^*_i=\Enc_{\pk_i}(\sk_{i-1})$ for all $i=1, \ldots, d$.
	
	\item Output $\pk = (\pk_0, (\pk_1, c^*_1), \ldots, (\pk_d, c^*_d))$ and $\sk = \sk_d$.
\end{enumerate}
For our purposes, we will assume that the random tape that is being used for the $\SubKeyGen$ executions is generated using a pseudorandom function. That is, the random tape of $\KeyGen$ is used as a seed for a PRF, and for the $i$th execution of $\SubKeyGen$ we use a random tape that is derived by applying a PRF on $i$.
\end{definition}
For the sake of completeness we note that the decryption algorithm of the bootstrapped scheme is the same as that of the base scheme, and that for the sake of encryption only $\pk_0$ is needed.

\subsection{Point Functions and Compute-and-Compare Functions}\label{sec:prelim-point-cc}
The class of compute-and-compare functions, as well as its subclass of point functions, plays an important role in this work. In this section we define these function classes.

\begin{definition}[Point function] Let $y \in \{0,1\}^n$.
	The point function \pf{y} is defined by
	\begin{align}
	\pf{y}(x) :=
	\begin{cases}
	1 &\text{if } x = y\\
	0 &\text{otherwise.}
	\end{cases}
	\end{align}
\end{definition}

\noindent The value $y$ is called the \emph{target value}. Point functions are a special type of compute-and-compare function, where the function $f$ is the identity:

\begin{definition}[Compute-and-compare function] \label{def:compute-and-compare}
  Let $f : \{0,1\}^{m} \to \{0,1\}^n$ and $y \in \{0,1\}^n$.
	The compute-and-compare function \cc{f}{y} is defined by
	\begin{align}
	\cc{f}{y}(x) :=
	\begin{cases}
	1 &\text{if } f(x) = y\\
	0 &\text{otherwise.}
	\end{cases}
	\end{align}
\end{definition}

\noindent One can also consider point functions or compute-and-compare functions with \emph{multi-bit output}: in that case, the function outputs either some string $z$ (instead of 1), or the all-zero string (instead of 0). We denote such functions with \mbpf{y}{z} and \mbcc{f}{y}{z}.

\subsection{Recovering the Input of a Quantum Circuit}\label{sec:prelim-input-recovery}
We will consider (efficient) quantum operations as (polynomial-size) circuits, consisting of the following set of basic operations: unitary gates from some fixed constant-size gate set, measurements in the computational basis, and initialization of auxiliary wires in the $\ket{0}$ state.

While unitary gates are always reversible by applying their transpose ($U^{\dagger}U = I$ for any unitary $U$), measurement gates may not be, as they can possibly collapse a state. However, we can effectively delay all measurements in a circuit $C$ until the very end, as follows. Define $U_C$ as the unitary that computes $C$ \emph{coherently}: that is, for every computational-basis measurement in $C$ on some wire $w$, $U_C$ performs a CNOT operation from $w$ onto a fresh auxiliary target wire initialized in the state $\ket{0}$. The circuit $C$ is now equivalent to the following operation: initialize all auxiliary target wires in the $\ket{0}$ state\footnote{If, apart from the targets of the aforementioned CNOTs, the circuit $C$ contains any other wires that are initialized in the $\ket{0}$ state inside the circuit, those wires are also considered part of the input of the unitary $U_C$. They should be initialized to $\ket{0}$ here as well.}, apply the unitary $U_C$, and measure all auxiliary target wires in the computational basis.

In this work, we will encounter circuits $C$ which, for specific inputs, yield a specific state in the computational basis with very high probability. In the proof of the following lemma, we specify how to use coherent computation in order to learn the output value while preserving the input quantum state.
\begin{lemma}\label{lem:input-recovery}
	Let $C$ be a quantum circuit. There exists an input-recovering circuit $\Crec$ such that for all inputs $\rho_{\inp}$, the following holds: if
	$
	\frac{1}{2} \trnorm{C(\rho_\inp) - \kb{x}} \leq \varepsilon
	$
	for some classical string $x$ and some $\varepsilon > 0$, then
	\[
	\frac{1}{2} \trnorm{\Crec(\rho_\inp) - \left(\rho_\inp \otimes \kb{x} \right) } \leq 2\sqrt{\varepsilon}.
	\]
\end{lemma}

The specification of \Crec is independent of the specific input state $\rho_\inp$. However, \Crec cannot necessarily recover \emph{all} possible inputs $\rho_\inp$, only those that lead to an almost-classical output.

The input-recovering circuit consists of running $C$ coherently, copying out the output register, and reverting the coherent computation of $C$. We formally prove \Cref{lem:input-recovery} in \Cref{app:proof-input-recovery}.

\section{FHE with Decomposable Public Keys}\label{sec:qfhe-decomposable}

For the purpose of our result in \Cref{sec:no-aux}, we will need to obfuscate a class of circuits that allow to (quantumly) homomorphically evaluate  operations of arbitrary polynomial depth. We nevertheless wish to rely only on leveled FHE for the sake of minimizing our assumptions. We therefore would like to define a class of circuits that are a priori polynomially bounded in size, but which are capable of encapsulating public-key generation of a leveled scheme for some depth $d$ that is not fixed a priori. Note that in a leveled scheme even the length of $\pk$ depends on $d$.

To this end, we define the notion of a scheme with \emph{decomposable} public key, which is defined below. Intuitively, in such a scheme, the public key can be generated by first generating a sequence of blocks, each of some size independent of $d$. These blocks can then be combined into the actual $\pk$ of the scheme. Crucially, the generation of the blocks can be done in parallel, and the complexity of generating each block (given the security parameter and the random tape) is independent of $d$. 
In other words, a decomposable public key can be generated on the fly, involving small ``chunks'' of computation that are independent of $d$.
Formally, we recall \Cref{def:levfhe} and define decomposability as follows.

\begin{definition}[Decomposable public key]\label{def:decomposable-key}
	A leveled (Q)FHE scheme has a \emph{decomposable public key} if there exists a polynomial $K=K(\lambda,d)$ and a polynomial-time deterministic function $\BlockGen(1^{\lambda},i,r,r')$ (where $r, r' \in \{0,1\}^{\lambda}$) that generates classical strings (``blocks") $c_i$ 
	such that the following holds:
	\begin{enumerate}
		\item \textbf{Correctness:} there exists a QPT $\Assemble$ such that for all $\lambda, d, r,$ and $r'$, letting $K=K(\lambda,d)$, it holds that
		\[
		\Assemble(c_0, c_1, c_2, \dots, c_K) = \pk,
		\]
		where $(\pk,\sk) = \KeyGen(1^{\lambda}, 1^d; r)$, and $c_i = \BlockGen(1^{\lambda}, i,r,r')$ for all $i$.
		\item \textbf{Simulatability:} there exists a QPT simulator $\simS$ such that for all $d$ and $r$,
		\[
		\simS(1^{\lambda}, \pk) \capprox (c_0, c_1, c_2, \dots, c_K),
		\]
		where $(\pk,\sk) = \KeyGen(1^{\lambda}, d,r)$, and the distribution on $(c_1, c_2, \dots, c_K)$ on the right-hand side is generated by selecting a uniformly random $r'$, and then for all $i$, setting $c_i = \BlockGen(1^{\lambda}, i, r, r')$.
	\end{enumerate}
\end{definition}

We emphasize that in \Cref{def:decomposable-key}, the randomness strings $r$ and $r'$ are the same for every run of $\BlockGen$. The reason for this choice is twofold. First, with our final goal in mind of obfuscating the $\BlockGen$ functionality, we want to avoid having to specify $K$ independent randomness strings (as that would considerably increase the size of the circuit to obfuscate). Second, most schemes require some form of correlation to exist between the different blocks. Thinking of $r$ and $r'$ as short random seeds for a PRF, this correlation can be realized by running the PRF on the same inputs (see, for example \Cref{sec:decomposable-bootstrapped}).

\subsection{Instantiation from Bootstrapped Schemes}\label{sec:decomposable-bootstrapped}

For bootstrapped schemes (see \Cref{def:bootfhe}), decomposability follows immediately by definition. In this case, we do not even need the extra randomness $r'$ and can simply set $\BlockGen(1^{\lambda},i,r,r')$ to be the process that evaluates $PRF_r(i-1)$ and $PRF_r(i)$ to generate random tapes for $\SubKeyGen$, uses this randomness to generate $(\sk_{i-1}, \pk_{i-1})$ and $(\sk_i, \pk_i)$, generates $c^*_i$ based on these values, and outputs $(\pk_i, c^*_i)$. In addition, for $i=0$, it simply computes $PRF_r(0)$, and uses the resulting randomness to generate $\pk_0$.

Existing QFHE schemes are based on bootstrapping~\cite{Mahadev18,Brakerski18}.
Without affecting security, we can assume that their randomness is sampled using a PRF as just described.

\begin{lemma}
	Bootstrapping-based leveled QFHE schemes with keys generated from a PRF have decomposable public keys.
\end{lemma}

\begin{proof}
	Define $K(\lambda,d) := d$, and $c_0 := \pk_0$. For $i > 0$, define the blocks $c_i$, which are generated by $\BlockGen(1^{\lambda}, i, r,r')$, as follows:
	\begin{align}
	c_i := (\pk_i, c^*_i = \Enc_{\pk_{i}}(\sk_{i-1})), \text{ where } (\pk_i, \sk_i) &\from \SubKeyGen(1^{\lambda} ; PRF_r(i)),\nonumber\\
	(\pk_{i-1}, \sk_{i-1}) &\from \SubKeyGen(1^{\lambda} ; PRF_r(i-1)).
	\end{align}
	Note that for public keys of this form, $\BlockGen$ does not make use of the additional randomness $r'$. 
	
	The assembly function $\Assemble(c_0, c_1, \dots, c_d)$ is a straightforward concatenation of all the blocks: $\Assemble(c_0, c_1, \dots, c_d) := (c_0, c_1, \dots, c_d)$.
	
	Simulatability as in \Cref{def:decomposable-key} is also easily satisfied: a simulator $\simS$, for a public key $\pk$ and index $i$, reads out the pair $(\pk_i, c^*_i)$. It can thereby exactly produce the list $(c_1, \dots, c_d)$.
\end{proof}

\subsection{Instantiation from Any Leveled (Q)FHE}

We now observe that we can instantiate the a (Q)FHE with decomposable public keys from any leveled scheme, even ones that are not bootstrapped. 
Decomposing the public key of a general QFHE scheme is done via garbled circuits~\cite{Yao86,App17}, as we will briefly outline here.
A block $c_i$ corresponds to a single garbled gate of the circuit for $\KeyGen$. That is, $\BlockGen(1^{\lambda}, i, r, r')$ returns a garbling of the $i$th gate\footnote{The total number of blocks, $K(\lambda,d)$, will be the number of gates in $\KeyGen(1^{\lambda}, 1^d, r)$. Since the number of gates is polynomial in $\lambda$, it suffices for the length of the PRF seed $r'$ to be linear in $\lambda$.} of $\KeyGen(1^{\lambda},d,r)$, using $r'$ as a PRF seed to generate sufficient randomness for the garbling. A separate block (e.g., $c_0$) contains the required encoding/decoding information to use the garbled circuit. To assemble the public key, a user concatenates all garbled gates, and evaluates the garbled circuit to obtain the output $\pk$. Conversely, by the privacy property of garbled circuits~\cite{BHR12}, a simulator given only the security parameter $\lambda$ and the output $\pk$ of the garbled circuit, can reproduce a garbled circuit that is indistinguishable from the actual garbled circuit. It can then return the gates of that simulated garbled circuit as the blocks $c_i$.

Any result relying on the decomposability of the public key of a non-bootstrapping based QFHE scheme of course also relies on any computational assumptions required for the security of the garbled-circuit construction.

\section{Impossibility with respect to Dependent Auxiliary Information}\label{sec:aux}
In this section, we show impossibility of virtual-black-box quantum obfuscation of classical point functions under dependent auxiliary information. It sets the stage for our main result, \Cref{thm:impossibility}, where we incorporate the auxiliary information into the circuit, constructing a circuit class which is unobfuscatable even without the presence of any auxiliary information. Although the result in the current section is perhaps less surprising, the proof contains the most important technical details of our work.

The impossibility result requires two cryptographic primitives, both of which can be built from the hardness of LWE~\cite{Mahadev18,WZ17,GKW17}: (1) quantum fully homomorphic encryption with classical client-side operations (see \Cref{sec:prelim-qfhe}), and (2) classical vbb obfuscation of compute-and-compare functions. Our result therefore holds under the assumption that LWE is hard. The least favorable LWE parameters are required for the obfuscation of compute-and-compare functionalities, and are discussed in \Cref{sec:classical-obfuscation-of-CC}.

In \Cref{sec:classical-obfuscation-of-CC}, we describe the classical obfuscator for compute-and-compare functions that we use. We will apply the construction from~\cite{WZ17,GKW17} to a specific class of compute-and-compare functions with a specific type of auxiliary information. In \Cref{sec:aux-unobfuscatable-class}, we use this specific application to define a class of circuits and auxiliary-information strings that is unobfuscatable in the quantum vbb sense. The impossibility proof follows in \Cref{sec:aux-impossibility-proof}.

\subsection{Classical obfuscation of compute-and-compare functions}\label{sec:classical-obfuscation-of-CC}
The works of \cite{WZ17,GKW17} showed that under the assumption that LWE (with polynomial dimension and exponential modulus in the security parameter $\lambda$) is hard, it is possible to classically obfuscate compute-and-compare functions~\cite{WZ17,GKW17}. We will write ``\LWEWZ" to denote their specific variant of the LWE assumption. We note that LWE is known to be at least as hard as worst-case lattice problems \cite{Regev05,PRSD17}. In particular, the aforementioned parameter regime \LWEWZ{} translates to the worst-case hardness of the Gap Shortest Vector Problem (GapSVP) with sub-exponential approximation factor (in the dimension of the lattice).
There is currently no known super-polynomial quantum speedup for GapSVP, and the best known quantum (and classical) algorithms require sub-exponential running time.

The works of \cite{WZ17,GKW17} achieve so-called distributional virtual-black-box obfuscation of functions $\cc{f}{y}$, assuming that the target value $y$ has sufficient pseudo-entropy given a description of the function $f$. The obfuscation is even secure in the presence of (dependent) auxiliary information, so long as the pseudo-entropy of the target value remains high, even conditioned on this auxiliary information.

In our construction, we provide a classically-obfuscated compute-and-compare function as auxiliary information to a quantum obfuscation. We will require that the target value of the compute-and-compare function is sufficiently random, even given the rest of the auxiliary information (including the quantum obfuscation).

More specifically, for any IND-secure public-key encryption scheme $(\KeyGen,\Enc,\Dec)$, fixed bit string $\alpha$, and a classical obfuscation procedure $\O(\cdot)$, define a distribution ensemble $\{D_{\lambda}^{\alpha,d}\}_{\lambda \in \mathbb{N}}$ that samples
\begin{align}
(\pk, \cipha, \obfsk) \from D_{\lambda}^{\alpha,d} 
\quad \text{ as } \quad &(\pk, \sk) \from \KeyGen(1^{\lambda},1^d), \nonumber\\
&\cipha \from \Enc_{\pk}(\alpha), \nonumber\\
&\beta \from_R \{0,1\}^{\lambda}\, ,\nonumber\\
&\obfsk \from \O\left(\cc{\Dec_{\sk}}{\beta}\right),\label{eq:distribution}
\end{align}
where $\cc{\Dec_{\sk}}{\beta}$ is a compute-and-compare function as in~\Cref{def:compute-and-compare}.
For each $\alpha$ and $\lambda$, the target value $\beta$ is chosen independently from all other information: its pseudo-entropy is $\lambda$, even conditioned on $\pk$, $\cipha$ and $\Dec_\sk$. Therefore, there exists an obfuscation procedure for this class of compute-and-compare programs that has distributional indistinguishability in the following sense:

\begin{lemma}
	[{Application of~\cite[Theorem 5.2]{WZ17}}]
	\label{lem:security-of-CC-obfuscation}
	Under the \LWEWZ{} assumption, there exists a classical obfuscation procedure $\ccobf{\cdot}$ and a (non-uniform) simulator $\simS$ such that for all $\alpha$ and $d$,
	\begin{align}
	(\pk, \cipha, \obfsk) \capprox (\pk, \cipha, \simS(1^{\lambda}, \params)),
	\end{align}
	where $(\pk, \cipha, \obfsk) \from D_{\lambda}^{\alpha,d}$ using $\ccobf{\cdot}$ as the obfuscation procedure $\O(\cdot)$, and $\params$ is some information that is independent of $\sk$ and $\beta$ (e.g., it may contain the size of the circuit and/or $\lambda$).
\end{lemma}

In the rest of this work, $\ccobf{\cdot}$ will implicitly be the obfuscation procedure used in the distributions $D^{\alpha,d}_{\lambda}$.

We note that the proofs in \cite{WZ17,GKW17} showed a classical reduction from the hardness of distinguishing the aforementioned distributions to the hardness of solving LWE. We note that proofs by (either Karp or Turing) classical polynomial-time reduction from $A$ to $B$ implies that any solver for $A$ can be translated into a solver for $B$ with comparable complexity, in particular if the solver for $A$ runs in quantum polynomial time then so will the resulting solver for $B$.

As a consequence of \Cref{lem:security-of-CC-obfuscation}, we show that it is hard to guess the value of $\alpha$, given only a ciphertext $\cipha$ for $\alpha$, and an obfuscation of the compute-and-compare function. Intuitively, since the information $\alpha$ is completely independent of the target value $\beta$, the obfuscation effectively hides the secret key $\sk$ that would be necessary to learn $\alpha$.

\begin{lemma}\label{lem:security-of-CC-with-side-information}
	Under the \LWEWZ{} assumption, there exists a negligible function $\negl{\cdot}$ such that for any QPT algorithm $A$ and any $d$,
	\begin{align}
	\Pr[A(\pk, \cipha, \obfsk) = \alpha] \leq \negl{\lambda}.
	\end{align}
	Here, the probability is over $\alpha \from_R \{0,1\}^{\lambda}$, $(\pk, \cipha, \obfsk) \from D_{\lambda}^{\alpha,d}$, and the execution of $A$.
\end{lemma}

\begin{proof}
	The result follows almost directly from \Cref{lem:security-of-CC-obfuscation}, except that we want to bound the probability that $A$ outputs the multi-bit string $\alpha$, whereas \Cref{lem:security-of-CC-obfuscation} only deals with algorithms with a single-bit output.
	
	To bridge the gap, define an algorithm $A'_{\alpha}$ that runs $A$ on its input, and compares the output of $A$ to $\alpha$: if they are equal, $A'_{\alpha}$ outputs 1; otherwise, it outputs 0.
	
	For any \emph{fixed} value of $\alpha$, we have
	
	\begin{align}
	\Pr[A(\pk, \cipha, \obfsk) = \alpha] &= \Pr[A'_{\alpha}(\pk, \cipha, \obfsk) = 1]\\
	&\stackrel{(*)}{\approx} \Pr[A'_{\alpha}(\pk, \cipha, \simS(1^{\lambda},\params)) = 1]\\
	&= \Pr[A(\pk, \cipha, \simS(1^{\lambda},\params)) = \alpha].
	\end{align}
	
	The approximation (*) follows from \Cref{lem:security-of-CC-obfuscation}, and holds up to a difference of $\negl{\lambda}$.
	
	To complete the proof, note that $\simS(1^{\lambda}, \params)$ depends neither on $\alpha$ nor on $\sk$. Thus, randomizing over $\alpha$ again, and invoking privacy of the encryption, we get
	\begin{align}
	\Pr[A(\pk, \cipha, \obfsk) = \alpha] \approx \Pr[A(\pk, \cipha, \simS(1^{\lambda},\params)) = \alpha] \leq \negl{|\alpha|} = \negl{\lambda}.
	\end{align}
\end{proof}
We have thus established that, even in the presence of an obfuscated compute-and-compare function that depends on the secret key, encryptions remain secure (in the one-way sense). For this security to hold, it is important that the target value $\beta$ is sufficiently independent of the plaintext $\alpha$.

\subsection{An unobfuscatable circuit class}\label{sec:aux-unobfuscatable-class}
In this subsection, we define the class of circuits and auxiliary-information strings that we will prove unobfuscatable. Like in~\cite{BGI+01}, we will exploit the idea that access to an object (circuit or quantum state) that allows the evaluation of a function is more powerful than mere black-box access to the functionality: in particular, it allows to evaluate the function homomorphically. For this argument to work, it is important that the function is not easily learnable through black-box access. We will use point functions, as in~\cite{BGI+01}: with black-box access only, it is hard to tell the difference between a point function and the all-zero function \zf{\lambda}, that accepts inputs of length $\lambda$, and always returns $0^{\lambda}$.

Consider the class $\Cpoint_{\lambda,d} \cup \Czero_{\lambda,d}$ of circuits plus auxiliary information, where

\begin{align}
\Cpoint_{\lambda,d} := \{(\mbpf{\alpha}{\beta}, &(\pk, \cipha,\obfsk)) \mid \alpha \in \{0,1\}^{\lambda}, (\pk, \cipha,\obfsk) \in \text{supp}(D_{\lambda}^{\alpha,d})\},\\
\Czero_{\lambda,d} := \{(\zf{\lambda}, &(\pk, \cipha,\obfsk)) \mid \alpha \in \{0,1\}^{\lambda}, (\pk, \cipha,\obfsk) \in \text{supp}(D_{\lambda}^{\alpha,d})\}.
\end{align}
The class $\Cpoint_{\lambda,d}$ contains all $\lambda$-bit point functions, together with an encryption of the point input $\alpha$, a public key that enables evaluation of circuits up to depth $d$, and a function that checks whether a ciphertext decrypts to the target value $\beta$. $\Czero_{\lambda,d}$ contains the all-zero function \zf{\lambda} (which is itself a point function), but still with auxiliary information for the possible values of $\alpha$ and $\beta$.

Suppose that some quantum obfuscation $(\qobf{\cdot}), \intJ)$ exists. We define a QPT algorithm $A$, which expects an obfuscation $\rho = \qobf{\mbpf{\alpha}{\beta}}$ (or $\qobf{\zf{\lambda}}$), together with the classical auxiliary information $\aux = (\pk, \cipha, \obfsk)$. On general inputs $\rho$ and $\aux = (\key, \ctxt, \obf)$ of this form, let $A$ do as follows:

\begin{enumerate}
	\item Run $\QFHE.\Eval_{\key}(\intJ, \Enc_{\key}(\rho) \otimes \kb{\ctxt})$ to homomorphically evaluate the interpreter algorithm $\intJ$. If $\rho = \qobf{\mbpf{\alpha}{\beta}}$, $\key = \pk$, and $\ctxt = \cipha$, then this step results in an encryption of $\beta$ with high probability. If $\rho = \qobf{\zf{\lambda}}$, $\key = \pk$, $\ctxt = \cipha$, and $d$ is at least the depth of $\J$, then it results in an encryption of $0^{\lambda}$. Note that we use classical and quantum ciphertexts for the QFHE scheme interchangeably here: see \Cref{sec:prelim-qfhe} for a justification.
	\item Run $\obf$ on the output of the previous step. If $\obf = \obfsk$, this will indicate whether the previous step resulted in a ciphertext for $\beta$ or not.
\end{enumerate}
The above algorithm $A$ will almost certainly output 1 when given an element from $\Cpoint_{\lambda,d}$ for a sufficiently high value of $d$, because of the functional equivalence of the two obfuscations and the correctness of the homomorphic evaluation. Similarly, when given an element from $\Czero_{\lambda,d} - \Cpoint_{\lambda,d}$, it will almost certainly output 0. Formally, for all $\alpha, \beta \in \{0,1\}^{\lambda} - \{0^{\lambda}\}$, and $d$ at least the depth of $\J$,
\begin{align}
\Pr\left[A(\qobf{\mbpf{\alpha}{\beta}}, \pk, \cipha, \obfsk) = 1\right] &\geq 1 - \negl{\lambda},\label{eq:aux-a-on-point-function}\\
\Pr\left[A(\qobf{\zf{\lambda}}, \pk, \cipha, \obfsk) = 1\right] &\leq \negl{\lambda}.\label{eq:aux-a-on-zero-function}
\end{align}
The vastly different output distribution of $A$ when given an obfuscation of a point function versus the zero function are due the fact that $A$ has an actual representation, $\rho$, of the function to feed into the interpreter $\intJ$. In the proof in the next subsection, we will see that a simulator, with only black-box access to these functionality, will not be able to make that distinction.

\subsection{Impossibility proof}\label{sec:aux-impossibility-proof}
We are now ready to state and prove the impossibility theorem for quantum obfuscation of classical circuits with dependent auxiliary input. We reiterate that the two assumptions (quantum FHE and compute-and-compare obfuscation) can be realized under the \LWEWZ{} assumption.

Define $\Cpoint_{\lambda} := \bigcup_{d \in [2^{\lambda}]} \Cpoint_{\lambda,d}$, and similarly $\Czero_{\lambda} := \bigcup_{d \in [2^{\lambda}]} \Czero_{\lambda,d}$.

\begin{theorem}[Impossibility of quantum obfuscation w.r.t.\ auxiliary input]\label{thm:impossiblity-auxiliary}
	Suppose that a classical-client quantum fully homomorphic encryption scheme \QFHE exists that satisfies \Cref{def:qfhe}, and a classical obfuscation procedure \ccobf{\cdot} for compute-and-compare functionalities exists that satisfies \Cref{lem:security-of-CC-obfuscation}. Then any (not necessarily efficient) quantum obfuscator $(\qobf{\cdot},\intJ)$ for the class $\Cpoint_{\lambda} \cup \Czero_{\lambda}$ satisfying conditions 1 (polynomial expansion) and 2 (functional equivalence) from \Cref{def:quantum-vbb} cannot be virtual black-box under the presence of classical dependent auxiliary input, i.e., cannot satisfy condition 3 from \Cref{def:quantum-vbb} where both $\advA$ and $\simS$ get access to a classical string $\aux$ (which may depend on $C$). 
\end{theorem}

It may seem that the class $\Cpoint_{\lambda} \cup \Czero_{\lambda}$, consisting of point functions, is classically obfuscatable using $\ccobf{\cdot}$ from~\cite{WZ17,GKW17}. That obfuscation is secure if $\alpha$ (wich is the target value if we view $\mbpf{\alpha}{\beta}$ as the multi-bit output compute-and-compare function $\mbcc{\textsf{id}}{\alpha}{\beta}$) is unpredictable given the auxiliary information $\aux = (\pk, \cipha, \obfsk)$. On the surface, that seems to be the case: only an encryption of $\alpha$ is available in the auxiliary information. However, the secret key $\sk$ is present as part of the compute-and-compare function $\cc{\Dec_{\sk}}{\beta}$. That function is obfuscated, but the obfuscation is not secure in the presence of (an obfuscation of) $\mbpf{\alpha}{\beta}$. Thus, the obfuscation result from~\cite{WZ17,GKW17} \emph{almost} applies to the class $\Cpoint_{\lambda} \cup \Czero_{\lambda}$, but not quite. Hence we are able to prove impossibility of obfuscating it, which we do below.

\begin{proof}
	The proof structure is similar to~\cite{BGI+01}, and is by contradiction: assume that a quantum obfuscation $(\qobf{\cdot}, \intJ)$ for the class $\Cpoint_{\lambda} \cup \Czero_{\lambda}$ does exist that satisfies all three conditions. We will show that the output distribution of the algorithm $A$ defined in \Cref{sec:aux-unobfuscatable-class} is approximately the same for every element of the class, contradicting \Cref{eq:aux-a-on-point-function,eq:aux-a-on-zero-function}.
	
	By the assumption of the existence of a secure quantum obfuscation $(\qobf{\cdot}, \intJ)$, there exists a simulator $\simS$ such that
	
	\begin{align}
	\left|
	\Pr[A(\qobf{\mbpf{\alpha}{\beta}}, \aux) = 1]
	-
	\Pr[\simS^{\mbpf{\alpha}{\beta}}(1^{\lambda}, \aux) = 1]\label{eq:aux-qobf-application-point}
	\right|
	&\leq
	\negl{\lambda} \text{, and}
	\\
	\left|
	\Pr[A(\qobf{\zf{\lambda}}, \aux) = 1]
	-
	\Pr[\simS^{\zf{\lambda}}(1^{\lambda}, \aux) = 1]\label{eq:aux-qobf-application-zero}
	\right|
	&\leq
	\negl{\lambda}.
	\end{align}
	Here, the probability is taken over $\alpha \from_R \{0,1\}^{\lambda}$ and $\aux = (\pk, \cipha, \ccobf{\cc{\Dec_\sk}{\beta}}) \from D_{\lambda}^{\alpha,q}$ for $q$ the depth of the interpreter circuit $\J$. Note that $\simS$ does not depend on $\alpha$, $\beta$, $\sk$, or $\pk$.
	
	In the remainder of this proof we show that for any $\simS$ (independent of $\alpha$, $\beta$, $\sk$, $\pk$),
	\begin{align}
	\left|
	\Pr[\simS^{\mbpf{\alpha}{\beta}}(1^{\lambda}, \aux) = 1]
	-
	\Pr[\simS^{\zf{\lambda}}(1^{\lambda}, \aux) = 1]
	\right|
	\leq
	\negl{\lambda},\label{eq:aux-sim-small-difference}
	\end{align}
	from which it can be concluded that
	\begin{align}
	\left|
	\Pr[A(\qobf{\mbpf{\alpha}{\beta}}, \aux) = 1]
	-
	\Pr[A(\qobf{\zf{\lambda}}, \aux) = 1]
	\right|
	\leq
	\negl{\lambda}.\label{eq:aux-a-difference-small}
	\end{align}
	Since \Cref{eq:aux-a-on-point-function,eq:aux-a-on-zero-function} imply that this difference must be at least $1 - \negl{\lambda}$, \Cref{eq:aux-a-difference-small} yields a contradiction.
	
	To show that Equation~\eqref{eq:aux-sim-small-difference} holds, i.e., to bound the difference in output probabilities of $\simS$ when given an oracle for \mbpf{\alpha}{\beta} versus an oracle for \zf{\lambda}, we employ the one-way to hiding theorem as it is stated in~\cite[Theorem 3]{AHU19}. It says that there exists a QPT algorithm $B$ such that
	\begin{align}
	\left|
	\Pr[\simS^{\mbpf{\alpha}{\beta}}(1^{\lambda}, \aux) = 1]
	-
	\Pr[\simS^{\zf{\lambda}}(1^{\lambda}, \aux) = 1]
	\right|
	\leq
	2d' \cdot \sqrt{\Pr[B^{\zf{\lambda}}(1^{\lambda}, \aux) = \alpha]},\label{eq:aux-O2H-application}
	\end{align}
	where $d' = \poly{\lambda}$ is the query depth of $\simS$. However, by \Cref{lem:security-of-CC-with-side-information}, the probability that $B$ outputs $\alpha$ when given the auxiliary information $\aux =  (\pk, \cipha, \obfsk)$ is negligible in $\lambda$. Granting $B$ access to the zero-oracle and the additional input $1^{\lambda}$ does not increase this probability, since the value of $\lambda$ can already be deduced from $\aux$.
	
	We can thus conclude that the difference in Equation~\eqref{eq:aux-O2H-application} is negligible, and Equation~\eqref{eq:aux-sim-small-difference} holds, as desired.
\end{proof}

\noindent We end this section with a few remarks: we describe some variants and generalizations of \Cref{thm:impossiblity-auxiliary} which almost immediately follow from the presented proof.

\begin{remark}
	The proof for \Cref{thm:impossiblity-auxiliary} also works if we replace $\ccobf{\cc{\Dec_{\sk}}{\beta}}$ inside the distributions $D_{\lambda}^{\alpha,d}$ with $\qobf{\cc{\Dec_{\sk}}{\beta}}$, the quantum obfuscation we get from the assumption. This adaptation renders a quantum obfuscator for point functions impossible with respect to dependent auxiliary \emph{quantum} input: a slightly weaker statement, but it does not require the existence of a classical obfuscator for compute-and-compare programs. In particular, the required LWE parameters are better, because we only need the assumption of quantum fully homomorphic encryption.
\end{remark}

\begin{remark}
	Even a quantum obfuscator $(\qobf{\cdot},\intJ)$ for $\Cpoint_{\lambda} \cup \Czero_{\lambda}$ with non-negligible errors in the functional equivalence and/or the virtual-black-box property would lead to a contradiction in the proof of \Cref{thm:impossiblity-auxiliary}. Concretely, let $\varepsilon_f$ denote the error for functional equivalence, and $\varepsilon_s$ denote the error for security in the virtual-black-box sense  (they are both $\negl{|C|} = \negl{\lambda}$ in \Cref{def:quantum-vbb}). The impossibility proof works for any values of $\varepsilon_f, \varepsilon_s$ such that $\varepsilon_f + \varepsilon_s \leq \frac{1}{2} - \frac{1}{\poly{\lambda}}$. So in particular, even a quantum obfuscator with small constant (instead of negligible) errors in both conditions cannot exist.
\end{remark}

\section{Impossibility without Auxiliary Information}\label{sec:no-aux}
In this section, we will show that quantum virtual-black-box obfuscation of classical circuits is impossible even when no auxiliary information is present. We will rely heavily on the class constructed in \Cref{sec:aux}, essentially showing how the auxiliary information can be absorbed into the obfuscated circuit. As a result, the unobfuscatable circuit class itself becomes perhaps less natural, but still consists of classical polynomial-size circuits. %
Thus, our theorem implies impossibility of quantum vbb obfuscation of the class of all efficient classical circuits.

We would like to consider circuits of the following form:
\begin{align}
C_{\alpha,\beta, \aux}(b,x) &:=
\begin{cases}
\aux = (\pk,\cipha,\obfsk) &\text{if } b = 0\\
\mbpf{\alpha}{\beta}(x) &\text{if } b = 1,
\end{cases}
\end{align}
where $(\pk,\cipha,\obfsk)$ is generated from $D_{\lambda}^{\alpha,d}$, as in \Cref{sec:aux}. The input bit $b$ is a choice bit: if it is set to 1, the function \mbpf{\alpha}{\beta} (or \zf{\lambda}) is evaluated on the actual input $x$, whereas if it is set to 0, the auxiliary information is retrieved.

The idea would then be to retrieve the auxiliary information, followed by a homomorphic evaluation of the branch for $b = 1$. There is a problem with this approach, however: since the auxiliary information $\aux$ contains the public evaluation key $\pk$, the circuit $C$ grows with $d$, which affects the length of \pk. But as the circuit grows, a (non-circularly-secure) QFHE scheme may require a larger \pk to perform all evaluation steps.

To get around this issue, the unobfuscatable circuit will generate the public key step-by-step, in a construction inspired by~\cite{CLTV15}. We will assume that the public key of the leveled QFHE scheme is decomposable in the sense of \Cref{def:decomposable-key}.

Given a scheme with a decomposable public key, we redefine the unobfuscatable circuit class as follows. Instead of returning the entire public key at once, the circuit allows the user to request individual blocks $c_i$, up to some depth $d$. An honest user can run the circuit $K+1 = K(d,\lambda)+1$ times to obtain $\pk = \Assemble(c_0, c_1, \dots, c_K)$. The depth $d$ will not be fixed a priori, although it will be (exponentially) upper bounded: the circuit will only be able to handle inputs $i$ where $|i| \leq \lambda$. Thus, only up to $2^{\lambda}$ blocks $c_i$ can be retrieved.

The circuit class we consider in this section consists of circuits of the following form:
\begin{align}
\hat{C}_{\alpha,\beta, d, r, r', \cipha, \obfsk}(b,x) &:=
\begin{cases}
(\cipha,\obfsk) &\text{if } b = 0,\\
\BlockGen(1^{\lambda}, x, r, r') &\text{if } b = 1 \text{ and } x \leq K(d,\lambda),\\
\bot &\text{if } b = 1 \text{ and } x > K(d,\lambda),\\
\mbpf{\alpha}{\beta}(x) &\text{if } b = 2.
\end{cases}
\end{align}
or
\begin{align}
\hat{C}'_{\alpha,\beta, d, r, r', \cipha, \obfsk}(b,x) &:=
\begin{cases}
(\cipha, \obfsk) &\text{if } b = 0,\\
\BlockGen(1^{\lambda}, x, r, r')  &\text{if } b = 1 \text{ and } x \leq K(d,\lambda),\\
\bot &\text{if } b = 1 \text{ and } x > K(d,\lambda),\\
\zf{\lambda}(x) &\text{if } b = 2.
\end{cases}
\end{align}
The first input $b$ is now a choice trit: depending on its value, a different branch of the circuit is executed.

We alter the distribution $D^{\alpha,d}_{\lambda}$ from \Cref{eq:distribution}, so that it does not explicitly generate the public key anymore. That information is now generated on-the-fly by setting $b=1$. The public and secret key are deterministically computed using $r$ to generate the auxiliary information $(\cipha,\obfsk)$ for $b = 0$. Consider the distribution ensemble $\{D^{\alpha,d,r}_{\lambda}\}_{\lambda \in \mathbb{N}}$, where
\begin{align}
(\cipha, \obfsk) \from D_{\lambda}^{\alpha,d,r} 
\quad \text{ as } \quad 
&(\pk, \sk) = \KeyGen(1^{\lambda}, 1^{d}, r), \nonumber\\
&\cipha \from \Enc_{\pk}(\alpha), \nonumber\\
&\beta \from_R \{0,1\}^{\lambda} \, \nonumber\\
&\obfsk \from \cc{\Dec_{\sk}}{\beta}.\label{eq:distribution-prf}
\end{align}
Note that the value of $d$ does not influence the size of $\cipha$ or $\obfsk$ (and thereby the circuit size of $\hat{C}$ and $\hat{C}'$).

We can then define the following parametrized circuit classes:
\begin{align}
\Chpoint_{\lambda,d} := \{\hat{C}_{\alpha,\beta,d,r,r',\cipha,\obfsk} \mid \alpha \in \{0,1\}^{\lambda}, r, r'\in \{0,1\}^{\lambda}, (\cipha,\obfsk) \in \text{supp}(D_{\lambda}^{\alpha,d,r})\},\\
\Chzero_{\lambda,d} := \{\hat{C}'_{\alpha,\beta,d,r,r',\cipha,\obfsk} \mid \alpha \in \{0,1\}^{\lambda}, r,r' \in \{0,1\}^{\lambda}, (\cipha,\obfsk) \in \text{supp}(D_{\lambda}^{\alpha,d,r})\}.
\end{align}
Define the circuit class $\Chpoint_{\lambda} \cup \Chzero_{\lambda}$, where $\Chpoint_{\lambda} := \bigcup_{d \in [2^\lambda]} \Chpoint_{\lambda,d}$ and similarly $\Chzero_{\lambda} := \bigcup_{d \in [2^\lambda]} \Chzero_{\lambda,d}$. Note that in all circuits in this class, the ``auxiliary information'' $(\cipha,\obfsk)$ is fixed. Hence, when the obfuscation of the compute-and-compare function is requested by setting $b = 0$, the circuit always returns the same obfuscation that depends on the same secret key $\sk$.

Similarly to the setting with auxiliary input, there exists a QPT algorithm $A'$ that has significantly different output distributions when given a circuit from $\Chpoint_{\lambda,d}$ versus a circuit from $\Chzero_{\lambda,d}$. Here, we define the algorithm $A'$ that is able to distinguish \emph{only} if it receives a circuit for $d = q$, where $q$ is the depth of the interpreter circuit. If $d < q$, then $A'$ will not be able to retrieve a long enough evaluation key, and will always output zero. However, for our impossibility result, a single value of $d$ on which $A'$ succeeds in distinguishing is sufficient. Note that we cannot define our circuit class to contain only circuits with $d = q$, since $q$ depends on the specific obfuscator/interpreter pair. 

On an input state $\rho$, we define $A'$ as follows:
\begin{enumerate}
	\item Run $\Jr(\rho, \kb{b=0} \otimes \kb{0^{\lambda}})$, where $\Jr$ is the input-recovering version of the interpreter circuit (see \Cref{lem:input-recovery}). If $\rho$ is an obfuscation of a circuit in $\Chpoint \cup \Chzero$, this will result in a state (negligibly close to) $\rho \otimes \kb{\cipha} \otimes \kb{\obfsk}$. Measure the second and third registers to obtain $\cipha$ and $\obfsk$.
	\item Let $q$ be the depth of the interpreter $\J$. Because the interpreter is efficient, $q = \poly{\lambda}$. Sequentially run $\Jr(\rho, \kb{b=1} \otimes \kb{i})$ for all $0 \leq i \leq K = {K(q,\lambda)}$ to obtain $(c_0, c_1, \dots, c_K)$, and compute the public evaluation key $\pk = \Assemble(c_0, c_1, \dots, c_K)$, suitable for homomorphic evaluations of up to depth $q$. Note that the key $\pk$ is only revealed in its entirety if the given circuit has parameter $d = q$.  If $d < q$, $A'$ will notice that $\bot$ is returned for some queries, and outputs 0 at this point.
	\item Run $\QFHE.\Eval_{\pk}(\intJ, \Enc_{\pk}(\rho) \otimes \kb{\Enc_{\pk}(b=2)} \otimes \kb{\cipha})$. Similarly to \Cref{sec:aux-unobfuscatable-class}, this will result in a ciphertext for $\beta$ (if $\rho$ was an obfuscation of a circuit in $\Chpoint_{\lambda}$) or a ciphertext for $0^{\lambda}$ (if $\rho$ was an obfuscation of a circuit in $\Chzero_{\lambda}$), provided that $d = q$.
	\item Run $\obfsk$ on the output of the previous step. Doing so will indicate whether the previous step resulted in a ciphertext for $\beta$ or not. If yes, output 1; otherwise output 0.
\end{enumerate}

Let $(O_Q(\cdot),\J)$ be an obfuscator. The algorithm $A'$, when given a random obfuscated circuit from $\Chpoint_{\lambda,q}$, will almost certainly output 1, where $q$ is the depth of $\J$. At the same time, an element from $\Chzero_{\lambda,q} - \Chpoint_{\lambda,q}$ will almost certainly result in the output 0. More formally, for all $\alpha, r \in \{0,1\}^{\lambda}$ and $d = q$,

\begin{align}
\Pr\left[A'(\qobf{\hat{C}_{\alpha,\beta,d,r,r',\cipha,\obfsk}}) = 1\right] &\geq 1 - \negl{\lambda},\label{eq:a-prime-on-point-function}\\
\Pr\left[A'(\qobf{\hat{C}'_{\alpha,\beta,d,r,r',\cipha,\obfsk}}) = 1\right] &\leq \negl{\lambda}.\label{eq:a-prime-on-zero-function}
\end{align}

The probability is taken over $D^{\alpha,d,r}_{\lambda}$, $r'$, and the internal randomness of $A'$. Compare these inequalities to \Cref{eq:aux-a-on-point-function,eq:aux-a-on-zero-function}.

We are now ready to state our main theorem.

\begin{theorem}[Impossibility of quantum obfuscation]\label{thm:impossibility}
	Suppose that a classical-client quantum fully homomorphic encryption scheme \QFHE exists that satisfies \Cref{def:qfhe,def:decomposable-key}, and a classical obfuscation procedure \ccobf{\cdot} for compute-and-compare functionalities exists that satisfies \Cref{lem:security-of-CC-obfuscation}. Then any (not necessarily efficient) quantum obfuscator $(\qobf{\cdot},\intJ)$ for the class $\Chpoint_{\lambda} \cup \Chzero_{\lambda}$ satisfying conditions 1 (polynomial expansion) and 2 (functional equivalence) from \Cref{def:quantum-vbb} cannot be virtual black-box, i.e., cannot satisfy condition 3 from \Cref{def:quantum-vbb}.
\end{theorem}

\begin{corollary}\label{cor:main-result}
	If the \LWEWZ{} assumption holds, the class of classical polynomial-size circuits cannot be quantum virtual-black-box obfuscated in the sense of \Cref{def:quantum-vbb}.
\end{corollary}

\begin{proof}[Proof of \Cref{thm:impossibility}]
	We again prove the statement by contradiction, assuming that there does exist an obfuscator $(\qobf{\cdot}, \intJ)$ that securely obfuscates $\Chpoint_{\lambda} \cup \Chzero_{\lambda}$. Let $q$ be the depth of $\J$, so that $K(q,\lambda)$ is the number of blocks $c_i$ of the evaluation key required by $A'$ to successfully distinguish between an element of $\Chpoint_{\lambda}$ and of $\Chzero_{\lambda}$.
	
	By the assumption that $(\qobf{\cdot},\intJ)$ is secure, there must exist a simulator $\simS_0$ such that for all $\alpha, r \in \{0,1\}^{\lambda}$ (and setting $d = q$),
	\begin{align}
	\left|
	\Pr[A'(\qobf{\hat{C}_{\alpha,\beta,q,r,r',\cipha,\obfsk}}) = 1]
	-
	\Pr[\simS_0^{\hat{C}_{\alpha,\beta,q,r,r',\cipha,\obfsk}}(1^{\lambda}) = 1]
	\right|
	&\leq
	\negl{\lambda},\\
	\left|
	\Pr[A'(\qobf{\hat{C}'_{\alpha,\beta,q,r,r',\cipha,\obfsk}}) = 1]
	-
	\Pr[\simS_0^{\hat{C}'_{\alpha,\beta,q,r,r',\cipha,\obfsk}}(1^{\lambda}) = 1]
	\right|
	&\leq
	\negl{\lambda}.
	\end{align}
	The probabilities are taken over $(\cipha, \obfsk) \from D_{\lambda}^{\alpha, d, r}$ and $r' \from_R \{0,1\}^{\lambda}$, and the internal randomness of $A'$ and $\simS_0$.
	
	The output distribution of $\simS_0$ can be exactly simulated by another simulator, $\simS_1$, that has access only to an oracle for $\mbpf{\alpha}{\beta}$ or $\zf{\lambda}$, and gets the auxiliary information $\pk$, $\cipha$, and $\obfsk$ as input. $\simS_1$ can simply run $\simS_0$, simulating each oracle query using its own oracle, auxiliary input, or a combination thereof. If (part of) the query of $\simS_0$ is for some block $c_i$, $\simS_1$ can use the decomposability of $\pk$ to compute the individual blocks. We formally show the existence of such an $\simS_1$ in \Cref{cor:oracle-to-input}.
	
	We can thus conclude that for all $\alpha, r \in \{0,1\}^{\lambda}$,
	\begin{align}
	\left|
	\Pr[A'(\qobf{\hat{C}_{\alpha,\beta,q,r,r',\cipha,\obfsk}}) = 1]
	-
	\Pr[\simS_1^{\mbpf{\alpha}{\beta}}(1^{\lambda}, \cipha,\obfsk, \pk) = 1]
	\right|
	&\leq
	\negl{\lambda},\\
	\left|
	\Pr[A'(\qobf{\hat{C}'_{\alpha,\beta,q,r,r',\cipha,\obfsk}}) = 1]
	-
	\Pr[\simS_1^{\zf{\lambda}}(1^{\lambda}, \cipha,\obfsk, \pk) = 1]
	\right|
	&\leq
	\negl{\lambda}.
	\end{align}
	Again, the probabilities are over $D_{\lambda}^{\alpha,d,r}$ and $r'$, $A'$, and $\simS_1$.
	
	However, by Equation~\eqref{eq:aux-sim-small-difference} in the proof of \Cref{thm:impossiblity-auxiliary}, the output distribution of $\simS_1$ can only differ negligibly between the two different oracles. Thus, we have
	\begin{align}
	\left|
	\Pr[A'(\qobf{\hat{C}_{\alpha,\beta,q,r,r',\cipha,\obfsk}}) = 1]
	-
	\Pr[A'(\qobf{\hat{C}'_{\alpha,\beta,q,r,r',\cipha,\obfsk}}) = 1]
	\right|
	\leq
	\negl{\lambda}.
	\end{align}
	This contradicts our observation in \Cref{eq:a-prime-on-point-function,eq:a-prime-on-zero-function} that on input $\hat{C}_{\alpha,\beta,q,r,r',\cipha,\obfsk}$, $A'$ will almost always output 1, whereas on input $\hat{C}'_{\alpha,\beta,q,r,r',\cipha,\obfsk}$, it will almost always output 0.
\end{proof}

\section*{Acknowledgements}
We thank Andrea Coladangelo, Urmila Mahadev and Alexander Poremba for useful discussions.
CS is supported by a NWO VIDI grant (Project No.\ 639.022.519). 
GA acknowledges support from the NSF under grant CCF-1763736, from the U.S. Army Research Office under Grant Number W911NF-20-1-0015, and from the U.S. Department of Energy under Award Number DE-SC0020312.
ZB is supported by the Binational Science Foundation (Grant No.\ 2016726), and by the European Union Horizon 2020 Research and Innovation Program via ERC Project REACT (Grant 756482) and via Project PROMETHEUS (Grant 780701).
Part of this work was done while the authors were attending \href{https://simons.berkeley.edu/programs/quantum2020}{The Quantum Wave in Computing} at the Simons Institute for the Theory of Computing.

\bibliographystyle{alphaarxiv}
\bibliography{references}

\appendix
\section{Proof of \Cref{lem:input-recovery}}\label{app:proof-input-recovery}
	
	\begin{proof}
		The input-recovering circuit \Crec will consist of running $C$ coherently, copying out the output register, and reverting the coherent computation of $C$. Suppose the circuit $C$ contains $k$ measurement gates, $\ell$ initializations of wires in the $\ket{0}$ state, and outputs of length $n$. Define \Crec as:
		\begin{enumerate}
			\item Run $U_C$ on input $\rho_\inp^{\reg{A_1}} \otimes \kb{0^{\ell}}^{\reg{A_2}} \otimes \kb{0^k}^{\reg{M}}$, where $U_C$ is the unitary that coherently executes $C$, $A = (A_1, A_2)$ is a register that contains the actual input and the auxiliary input $\ket{0}$ states for $C$, and $M$ is the register that contains the auxiliary wires for the coherent measurements.
			\item Copy the wires that are supposed to contain the output $C(\rho_\inp)$ into a register $Y$, initialized to $\kb{0^n}$, using $\CNOT$s. The source of the $\CNOT$s is a register $O$, the subregister of $A$ containing those output wires. Write $\overline{O}$ for the registers in $A$ that are not in $O$ (these wires are normally discarded after the execution of $C$). \label{step:obf:measure}
			\item Run $U_C^{\dagger}$ to recover the original input, and discard the registers $A_2$ and $M$.
		\end{enumerate}
		The behavior of $\Crec$ can be summarized as
		\begin{align}
		\Crec(\rho_{\inp}) = \Tr_{A_2M} \left[U_C^{\dagger} \CNOT^{\otimes n}_{O, Y} \left(U_C\left(\rho_{\inp}^{\reg{A_1}} \otimes \kb{0^{\ell + k}}^{\reg{A_2M}}\right)U_C^{\dagger} \otimes \kb{0}^{\reg{Y}}\right)\left(\CNOT^{\otimes n}_{O,Y}\right)^{\dagger} U_C\right].
		\end{align}
		To see that \Crec acts as promised, let $\rho_{\inp}$, $x$, and $\varepsilon$ be such that $\trnorm{C(\rho_{\inp}) - \kb{x}} \leq \varepsilon$. If $\varepsilon$ is small, the \CNOT in Step~\ref{step:obf:measure} does not create a lot of entanglement, since the control wires are (close to) the computational-basis state $\kb{x}$. The output is therefore (almost) perfectly copied out.
		
		More formally, note that $C(\rho_{\inp}) = \Tr_{\overline{O}M}\left[U_C(\rho_{\inp} \otimes \kb{0^{\ell + k}})U_C^{\dagger}\right]$. By Lemma A.1 in~\cite{ABC+19}, the closeness of $C(\rho_{\inp})$ and $\kb{x}$ implies that there exists a density matrix $\chi^{\reg{\overline{O}M}}$ such that
		\begin{align}
		\frac{1}{2} \trnorm{U_C(\rho_{\inp} \otimes \kb{0^{\ell + k}})U_C^{\dagger} \ \ - \ \  \kb{x}\otimes \chi^{\reg{\overline{O}M}}} \leq \sqrt{\varepsilon}.\label{eq:obf:sqrtdiff}
		\end{align}
		Next, we use the fact that a quantum map cannot increase the trace distance between two states to derive two inequalities from \Cref{eq:obf:sqrtdiff}.
		
		For the first inequality, we append $\kb{x}$ on both sides (into a separate $Y$ register):
		\begin{align}
		\frac{1}{2}\trnorm{U_C(\rho_{\inp} \otimes \kb{0^{\ell + k}})U_C^{\dagger} \otimes \kb{x}^{\reg{Y}} \ \ - \ \  \kb{x}\otimes \chi^{\reg{\overline{O}M}} \otimes \kb{x}^{\reg{Y}}} \leq \sqrt{\varepsilon}.\label{eq:obf:first-ineq}
		\end{align}
		
		For the second inequality, we instead append $\kb{0}$ into the $Y$ register, followed by $\CNOT$s from $O$ onto $Y$. Note that on the second term inside the trace norm, the effect is the same as before:
		\begin{align}
		\frac{1}{2}\trnorm{\CNOT^{\otimes n}_{O,Y}\left(U_C(\rho_{\inp} \otimes \kb{0^{\ell + k}})U_C^{\dagger} \otimes \kb{0}^{\reg{Y}}\right)\left(\CNOT^{\otimes n}_{O,Y}\right)^{\dagger} \ \ - \ \  \kb{x}\otimes \chi^{\reg{\overline{O}M}} \otimes \kb{x}^{\reg{Y}}} \leq \sqrt{\varepsilon}.\label{eq:obf:second-ineq}
		\end{align}
		Thus, by the triangle inequality, the left-hand terms inside the trace norms in \Cref{eq:obf:first-ineq,eq:obf:second-ineq} are $2\sqrt{\varepsilon}$-close. Applying the map $\Tr_{A_2M}\left[U_C^{\dagger} (\cdot) U_C\right]$ to both terms, which again does not increase the trace difference, we arrive at the desired statement:
		\begin{align}
		\frac{1}{2}\trnorm{\left(\rho_{\inp} \otimes \kb{x}\right) \ \ - \ \ \Crec(\rho_{\inp})} \leq 2\sqrt{\varepsilon}.
		\end{align}
	\end{proof}

\section{Auxiliary Lemmas for \Cref{thm:impossibility}}

\begin{lemma}\label{lem:classical-choice-oracle}
	Let $g : \{0,1\}^m \to \{0,1\}^n$ for $m, n \in \mathbb{N}$, and let $c \in \{0,1\}^n$. Let $f : \{0,1\} \times \{0,1\}^m \to \{0,1\}^n$ be defined by
	\begin{align}
	f(b,x) := 
	\begin{cases}
	c    &\text{if } b = 0\\
	g(x) &\text{if } b = 1.
	\end{cases}
	\end{align}
	Then for every QPT $A$, there exists a simulator $S$ such that for all $f, g$ of the form described above, and all input states $\rho$:
	\begin{align}
	\Pr[A^f(\rho) = 1] \ \ = \ \ \Pr[S^g(\rho, c) = 1].
	\end{align}
	
\end{lemma}

\begin{proof}
	Recall that since $A$ and $S$ are quantum algorithms, they access their oracles in superposition: that is, $A$ has access to the map defined by $\ket{x}\ket{z} \mapsto \ket{x} \ket{z \oplus f(x)}$, and $S$ has access to the map defined by $\ket{x}\ket{z} \mapsto \ket{x} \ket{z \oplus g(x)}$.
	The simulator $S$ runs $A$ on input $\rho$, and simulates any oracle calls to $f$ (on inputs registers $BX$ and output register $Z$) using two oracle calls to $g$. It only needs to prepare an auxiliary register in the state $\ket{0^n}$, and run the following circuit:
	
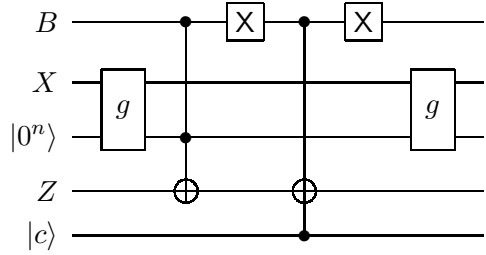
\[
	\Qcircuit @C=1em @R=1em {
		\lstick{B}         & \qw                & \ctrl{3} & \gate{\X} & \ctrl{3}  & \gate{\X} & \qw                & \qw \\
		\lstick{X}         & \multigate{1}{g} & \qw      & \qw       & \qw       & \qw       & \multigate{1}{g} & \qw \\
		\lstick{\ket{0^n}} & \ghost{g}        & \ctrl{1} & \qw       & \qw       & \qw       & \ghost{g}        & \qw \\
		\lstick{Z}         & \qw                & \targ    & \qw       & \targ     & \qw       & \qw                & \qw \\
		\lstick{\ket{c}}   & \qw                & \qw      & \qw       & \ctrl{-1} & \qw       & \qw                & \qw
	}
	\]
	To see that this circuit exactly simulates a query to $f$ on $BXZ$, consider an arbitrary query state
	\begin{align}
	\sum_i \alpha_i \ket{b_i, x_i}_{BX}\ket{z_i}_{Z}\ket{\varphi_i}_R,
	\end{align}
	where $R$ is some purifying register. The state on $BXZR$ (plus the two auxiliary registers containing $\ket{0^n}$ and $\ket{c}$) after the above circuit is executed, is equal to
	\begin{align}
	&\sum_i \alpha_i \ket{b_i, x_i}_{XB} \ket{0^n} \ket{z_i \oplus b \cdot g(x_i) \oplus (1-b) \cdot c}_Z \ket{c} \ket{\varphi_i}_R\\
	=&\sum_i \alpha_i \ket{b_i, x_i}_{XB} \ket{0^n} \ket{z_i \oplus f(x_i)}_Z \ket{c} \ket{\varphi_i}_R,
	\end{align}
	which is exactly the state that would result from a direct query to $f$.
\end{proof}

\begin{corollary}\label{cor:oracle-to-input}
	Let $\Chpoint_{\lambda}$ and $q$ be as in \Cref{sec:no-aux}. Then for any QPT $\simS_0$, there exists a QPT simulator $\simS_1$ such that for all $\alpha,r \in \{0,1\}^{\lambda}$,
	\begin{align}
	\left|
	\Pr[\simS_0^{\hat{C}_{\alpha,\beta,q,r,r',\cipha,\obfsk}}(1^{\lambda}) = 1] - 
	\Pr[\simS_1^{\mbpf{\alpha}{\beta}}(1^{\lambda}, \cipha,\obfsk, \pk) = 1]
	\right| \leq \negl{\lambda}.
	\end{align}
	A similar statement holds for circuits from $\Chzero_{\lambda}$.
\end{corollary}
\begin{proof}
	The statement is proven via an intermediate simulator $\simS_2$. This simulator is constructed by repeated application of \Cref{lem:classical-choice-oracle}, so that for all $\alpha,r$,
	\begin{align}
	\left|
	\Pr[\simS_0^{\hat{C}_{\alpha,\beta,q,r,r',\cipha,\obfsk}}(1^{\lambda}) = 1] - 
	\Pr[\simS_2^{\mbpf{\alpha}{\beta}}(1^{\lambda}, \cipha,\obfsk, c_0, c_1, c_2, \dots, c_K, \bot) = 1]
	\right| \leq \negl{\lambda},
	\end{align}
	where $K = K(q,\lambda)$ as in \Cref{def:decomposable-key}. 
	On the right-hand side, the probability is additionally over a random choice of $r'$ (resulting in the sequence $(c_0, c_1, c_2, \dots, c_K)$).
	
	Next, we apply the simulatability property of \Cref{def:decomposable-key}. It states that there exists a simulator $\simS_3$ that, given a public key, can generate the distribution over $(c_0, c_1, c_2, \dots, c_K)$ itself. Define
	\begin{align}
	\simS_1^{\mbpf{\alpha}{\beta}}(1^{\lambda}, \cipha,\obfsk, \pk) := \simS_2^{\mbpf{\alpha}{\beta}}(1^{\lambda}, \cipha,\obfsk, \simS_3(\pk), \bot),
	\end{align}
	and the corollary follows.
\end{proof}

\end{document}